\documentclass[journal]{IEEEtran}
\usepackage{graphicx}
\usepackage{bbm}
\usepackage{epsfig}
\usepackage{graphics,graphicx,amssymb,amsmath,verbatim}
\usepackage{amssymb}
\usepackage{mathrsfs}
\usepackage{amsfonts}

\newtheorem{corollary}{Corollary}

\newtheorem{proposition}{Proposition}
\newtheorem{remark}{Remark}

\newenvironment{proof}[1][Proof]{\noindent\textbf{#1.} }{\ \rule{0.5em}{0.5em}}
\ifCLASSINFOpdf
\else
\fi
\input{tcilatex}
\begin{document}

\title{A Divide and Conquer Approach to Cooperative Distributed Model
Predictive Control}
\author{He~Kong,~Stefano~Longo,~Gabriele~Pannocchia,~Efstathios~Siampis,
and~Lilantha~Samaranayake \thanks{This work has been supported by the \textquotedblleft Developing FUTURE
Vehicles\textquotedblright\ project of the Engineering and Physical Sciences
Research Council under the UK Low Carbon Vehicles Integrated Delivery
Programme.}\thanks{%
He Kong was with the Advanced Vehicle Engineering Centre, Cranfield
University, United Kingdom; he is now with the Australian Centre for Field
Robotics, the University of Sydney, NSW, Australia, Email:
h.kong@acfr.usyd.edu.au. Stefano Longo and Efstathios Siampis are with the Advanced Vehicle Engineering Centre, Cranfield
University, College Road, Cranfield, United Kingdom. Email: {s.longo,
e.siampis}@cranfield.ac.uk. Gabriele Pannocchia is with
the Department of Civil and Industrial Engineering, University of Pisa,
Italy, Email: gabriele.pannocchia@unipi.it. Lilantha Samaranayake is with Department of Electrical and Electronic Engineering, University of Peradeniya, Sri Lanka, Email: lilantha@ee.pdn.ac.lk }}
\maketitle

\begin{abstract}
This paper is concerned with the design of cooperative distributed Model
Predictive Control (MPC) for linear systems. Motivated by the special
structure of the distributed models in some existing literature, we propose
to apply a state transformation to the original system and global cost
function. This has major implications on the closed-loop stability analysis
and the mechanism of the resultant cooperative framework. It turns out that
the proposed framework can be implemented without cooperative iterations
being performed in the local optimizations, thus allowing one to compute the
local inputs in parallel and independently from each other while requiring
only partial plant-wide state information. The proposed framework can also
be realized with cooperative iterations, thereby keeping the advantages of
the technique in the former reference. Under certain conditions, closed-loop
stability for both implementation procedures can be guaranteed a priori by
appropriate selections of the original local cost functions. The strengths
and benefits of the proposed method are highlighted by means of two
numerical examples.
\end{abstract}

%\author{Michael~Shell,~\IEEEmembership{Member,~IEEE,}
%        John~Doe,~\IEEEmembership{Fellow,~OSA,}
%        and~Jane~Doe,~\IEEEmembership{Life~Fellow,~IEEE}% <-this % stops a space

% <-this % stops a space

\markboth{}{Shell \MakeLowercase{\textit{et al.}}: A Divide and Conquer Approach to Cooperative Distributed Model
Predictive Control}



% Note that keywords are not normally used for peer-review papers.
\begin{IEEEkeywords}
Model predictive control, Distributed control, Cooperative control, Constrained control
\end{IEEEkeywords}

\IEEEpeerreviewmaketitle

\section{Introduction}

MPC relies on a dynamic model of the system of interest to predict its
behavior into the future, and solves, at each sampling instant, a {finite}
horizon optimization problem to determine an input sequence, while taking
the system's constraints into account \cite{Mayne2000}-\cite{Mayne2014}.
This paper is inspired by the recent rigorous development in decentralized
and distributed MPC \cite{Pannocchia2009}-\cite{Ling2012}. Depending on the
particular problem setups and associated solutions, existing distributed MPC
methods could be classified in ways such as decentralized or distributed,
cooperative or noncooperative \cite{Keviczky2006}-\cite{Cannon2015}.
Analysis and design of distributed MPC with network-induced effects has also
received a considerable amount of attention \cite{Parisini2011}-\cite%
{Shi2013}. For more detailed discussions on the properties of different
distributed MPC methods, one can refer to \cite[Chp.~6]{Madison2009}, \cite%
{Scattolini}-\cite{Pannocchia2015}.

This paper is concerned with the questions on whether and how the
organizational system architecture can be utilized in the design of
cooperative distributed MPC to facilitate the computation of local inputs
and reduce the communication burden with stability guarantee. The importance
of exploiting structural properties in synthesis and design of large-scale
systems has been demonstrated in a number of works in control \cite%
{Rotkowitz2006}-\cite{Ishizaki2014}. Another long-learned lesson is that
understanding of architectural issues in controller design are frequently
more important than optimizing within a given framework \cite[{Chp}. 1, pp.
10-18]{CSD}. A major motivation for developing the ideas in this paper is
close in spirit to the above observations. To illustrate the concept, we
adopt the problem setup in \cite{Rawlings2010}. We advocate the utilization
of a state transformation to the considered system and the original global
cost function so that the intrinsic coupling effects in the original
cooperative question are restructured and can be dealt with more effectively.

{The} proposed framework in this paper has the following advantages.
Firstly, the proposed method can be implemented %with or
without {cooperative} iterations, if certain conditions hold (see in
Propositions \ref{prop}-\ref{prop2}), with guarantee of global closed-loop
stability. {This} allows one to compute local inputs in parallel and
independently from each other, requiring only partial state information of
other subsystems (besides the subsystem matrices information). This is in
contrast to {the standard cooperative MPC formulations \cite{Rawlings2010}
in which} the computation of the local input requires plant-wide state
information and inputs of other local systems, or at least of an \emph{%
augmented local }system state \cite{Razzanelli2016}. Second, the proposed
framework can also be realized with iterations as in \cite{Rawlings2010},
thereby keeping the merits of the standard cooperative method with iterations%
{, e.g. convergence to the global centralized performance. }Given the impact
of the afore-mentioned state transformation on the original problem setup
and the features of the proposed framework, we term the method the \textit{%
Divide and Conquer }(D\&C)\textit{\ }approach.

The remainder of this paper is organized as follows. In Section \ref{lyap},
preliminaries are given. Section \ref{sgscopy} introduces the state
transformation and analyzes its impact on the solution to the cooperative {%
problem}. Section \ref{sgscopy2} contains stability analysis with state
feedback for both the centralized and the distributed cases. Section \ref%
{sgscopy3} presents extensions and comparison discussions. Section \ref%
{example} contains two simulation case studies. Section \ref{conclusion}
concludes the paper. \textbf{Notation}: Most notation we use is standard.
Matrices, whose dimensions are not stated, are assumed to be compatible for
algebraic operations. $\mathscr{I}_{s_{1}:s_{2}\text{ }}$denote the set of $%
\left\{ s_{1},s_{1}+1,\cdots ,s_{2}\right\} $. $\mathbf{U}^{N}$ denotes the
Cartesian product of set $\mathbf{U}$ for $N$ times. $diag(A,B)$ denotes a
block diagonal matrix with $A$ and $B$ as its block entries. $[v_{1},\cdots
,v_{n}]$ denotes $[v_{1}^{\mathrm{T}}\cdots v_{n}^{\mathrm{T}}]^{\mathrm{T}}$%
.

\section{\label{lyap}Preliminaries}

This section recalls the problem setup from \cite{Rawlings2010}. For
simplicity, we assume that the overall plant consists of only two
subsystems. The results will be extended to any finite number of subsystems
in Section \ref{sgscopy3}. Assume for $(i,j)$ $\in \mathscr{I}_{1:2}\times %
\mathscr{I}_{1:2},$ each subsystem $i$ is a collection of the linear
discrete-time models%
\begin{equation}
x_{ij}^{+}=A_{ij}x_{ij}+B_{ij}u_{j}\text{, \ }y_{i}=\sum\limits_{j\in %
\mathscr{I}_{1:2}}C_{ij}x_{ij}  \label{sub}
\end{equation}%
where $A_{ij}\in \mathbf{R}^{n_{ij}\times n_{ij}}$, $B_{ij}\in \mathbf{R}%
^{n_{ij}\times m_{j}}$, $C_{ij}\in \mathbf{R}^{p_{i}\times n_{ij}}$ are
constant system matrices; $u_{j}\in \mathbf{R}^{n_{ij}}$ denotes the effects
of input of subsystem $j$ on the states of subsystem $i$. Denote $%
x_{1}=[x_{11},x_{12}].$ By collecting the states of subsystem $1$ from (\ref%
{sub}), we can obtain:
\begin{equation}
x_{1}^{+}=A_{1}x_{1}+\overline{B}_{11}u_{1}+\overline{B}_{12}u_{2}\text{, \ }%
y_{1}=C_{1}x_{1}  \label{localsystem}
\end{equation}%
in which $A_{1}\in \mathbf{R}^{n_{1}\times n_{1}},$ $\overline{B}_{1j}\in
\mathbf{R}^{n_{1}\times m_{j}},$ $n_{1}=n_{11}+n_{12}$, $%
A_{1}=diag(A_{11},A_{12}),\overline{B}_{11}=[B_{11},0],\overline{B}%
_{12}=[0,B_{12}],C_{1}=\left[
\begin{array}{cc}
C_{11} & C_{12}%
\end{array}%
\right] .$ The model of subsystem 2 can be obtained similarly. The
plant-wide model then becomes%
\begin{equation}
x^{+}=Ax+B_{1}u_{1}+B_{2}u_{2}\text{, \ }y=Cx  \label{system}
\end{equation}%
where $x=[x_{1},x_{2}],$ $y=[y_{1},y_{2}],B_{1}=[\overline{B}_{11},\overline{%
B}_{21}],B_{2}=[\overline{B}_{12},\overline{B}_{22}],$ $A=diag(A_{1},A_{2}),$
and $C=diag(C_{1},C_{2})$. As remarked in \cite{Rawlings2010}, the model (%
\ref{system}) is potentially non-minimal. Denote $\mathbf{u}_{1}$ and $%
\mathbf{u}_{2}$ the local input sequences for the two subsystems along the
prediction horizon, respectively. We define the cost
\begin{equation}
\begin{array}{c}
V_{1}(x_{1}(0),\mathbf{u}_{1},\mathbf{u}_{2})=\sum\limits_{k=0}^{N-1}[x_{1}^{%
\mathrm{T}}(k)Q_{1}x_{1}(k)+u_{1}^{\mathrm{T}}(k)R_{1}u_{1}(k)] \\
+x_{1}^{\mathrm{T}}(N)P_{1}x_{1}(N)%
\end{array}
\label{localcost}
\end{equation}%
for subsystem 1 and $V_{2}(x_{2}(0),\mathbf{u}_{1},\mathbf{u}_{2})$ for
subsystem 2 similarly. The plant-wide cost function is $\mathscr{V}(x(0),%
\mathbf{u}_{1},\mathbf{u}_{2})=\rho _{1}V_{1}(x_{1}(0),\mathbf{u}_{1},%
\mathbf{u}_{2})+\rho _{2}V_{2}(x_{2}(0),\mathbf{u}_{1},\mathbf{u}_{2})$,
where $\rho _{1}$ and $\rho _{2}$ are positive real numbers. For simplicity,
denote $V_{1},V_{2},\mathscr{V}$ as the local and global cost functions,
respectively. $\mathscr{V}$ can be expressed in terms of the original
centralized model (\ref{system}) as:%
\begin{equation}
\mathscr{V}=\sum\limits_{k=0}^{N-1}[x^{\mathrm{T}}(k)Qx(k)+u^{\mathrm{T}%
}(k)Ru(k)]+x^{\mathrm{T}}(N)Px(N)  \label{originalcost}
\end{equation}%
with $u(k)=[u_{1}(k),u_{2}(k)]$, $Q=diag(\rho _{1}Q_{1},\rho
_{2}Q_{2}),R=diag(\rho _{1}R_{1},\rho _{2}R_{2}),P=diag(\rho _{1}P_{1},\rho
_{2}P_{2}).$ %As in \cite{Rawlings2010},
{We} assume that the local inputs must satisfy $u_{1}(k)\in \mathbf{U}%
_{1},u_{2}(k)\in \mathbf{U}_{2},$ for $k\in \mathscr{I}_{0:N-1}$, where both
$\mathbf{U}_{1}$ and $\mathbf{U}_{2}$ are compact and convex sets that
include the origin in their interior. For the case of decoupled input
constraints, the contribution of either local input on $\mathscr{V}$ cannot
be affected by the other {one} in terms of feasibility. However, both $V_{1}$
and $V_{2}$ are functions of $\mathbf{u}_{1}$ and $\mathbf{u}_{2}$ because
dynamics of both subsystems are dependent on either of the two local inputs.
Thus, the computation of the optimal input sequence $\mathbf{u}_{1}$ will
still be dependent on $\mathbf{u}_{2}$, and {viceversa}. Therefore, for $%
i\in \mathscr{I}_{1:2},$ the local optimization problems of computing $%
\mathbf{u}_{i}$ will be%
\begin{equation}
\min\limits_{\mathbf{u}_{i}}\mathscr{V}_{i}\text{ }s.t.\text{ (\ref{system}%
), }u_{i}(k)\in \mathbf{U}_{i},\text{ for }k\in \mathscr{I}_{0:N-1},
\label{localcom}
\end{equation}%
where, $\mathscr{V}_{i}=\sum\limits_{k=0}^{N-1}[x^{\mathrm{T}%
}(k)Qx(k)+u_{i}^{\mathrm{T}}(k)\rho _{i}R_{i}u_{i}(k)]+x^{\mathrm{T}%
}(N)Px(N).$ Depending on the method, other requirements have to be enforced
in (\ref{localcom}) for establishing closed-loop stability. For example, in
\cite{Rawlings2010}, the prediction horizon $N$ has to be sufficiently long
to zero the unstable system modes. {In other cases \cite%
{Limon,Razzanelli2016,Stewart2011}, the usual terminal penalty and terminal
region constraints are adopted. Closed-loop stability is typically proved by
adopting the suboptimal strategy in \cite{suboptimal}. }For $i\in \mathscr{I}%
_{1:2},$ denote%
\begin{equation}
\overline{A}_{i}=diag(A_{1i},A_{2i}),\text{ }\widetilde{B}%
_{i}=[B_{1i},B_{2i}],  \label{ABT}
\end{equation}%
with $A_{1i},A_{2i},B_{1i},B_{2i}$ being defined in (\ref{sub})-(\ref{system}%
). Obviously, we have $\overline{A}_{i}\in \mathbf{R}^{\overline{n}%
_{i}\times \overline{n}_{i}},\widetilde{B}_{i}\in \mathbf{R}^{\overline{n}%
_{i}\times m_{i}},$ where%
\begin{equation}
\overline{n}_{i}=n_{1i}+n_{2i}.  \label{SS}
\end{equation}%
For $i\in \mathscr{I}_{1:2}$, the following assumptions are made in \cite%
{Rawlings2010} to establish stability: (\textbf{A1}) $(\overline{A}_{i},%
\widetilde{B}_{i})$ is stabilizable; (\textbf{A2}) $P_{i}>0,R_{i}>0,Q_{i}%
\geq 0$; (\textbf{A3}) $(A_{i},Q_{i})$ and $(A_{i},C_{i})$ are detectable.

\section{\label{sgscopy}The D\&C Approach}

\subsection{Dividing the system by a state transformation}

From (\ref{system}), one can notice that the model is of a special structure
in the sense that each local input has only partial impact on both local
systems. To be specific, from (\ref{sub}), for $i=1$, one has $\overline{A}%
_{1}=diag(A_{11},A_{21}),$ $\widetilde{B}_{1}=[B_{11},B_{21}]$, i.e., $u_{1}$
($u_{2}$) can affect $x_{11}$ and $x_{21}$ ($x_{12}$ and $x_{22}$). Given
that $(\overline{A}_{i},\widetilde{B}_{i})$ is supposed to be stabilizable
in Assumption A1, there exists $K_{i}$ such that $\overline{A}_{i}+%
\widetilde{B}_{i}K_{i}$ is Schur stable, i.e., to stabilize the system with
state feedback, for computing $u_{1}$ ($u_{2}$), one only needs the
information related to $x_{11}$ and $x_{21}$ ($x_{12}$ and $x_{22}$). This
motivates us to separate the modes of the subsystems into two categories
that can only be affected by $u_{1}$ or $u_{2}$, respectively. As such, we
introduce a state transformation which renders a clearer inspection of the
inherent structures of the original cooperative question (\ref{localcom}).
Denote%
\begin{equation}
T=\left[
\begin{array}{cccc}
I_{n_{11}} & 0 & 0 & 0 \\
0 & 0 & I_{n_{12}} & 0 \\
0 & I_{n_{21}} & 0 & 0 \\
0 & 0 & 0 & I_{n_{22}}%
\end{array}%
\right] .  \label{texpress}
\end{equation}%
One then has $T^{\mathrm{T}}T=I_{n_{1}+n_{2}}$\ and $T^{-1}=T^{\mathrm{T}}$
with $n_{1}$ being defined in (\ref{localsystem}) and $n_{2}$ defined
similarly. Denoting and substituting%
\begin{equation}
x=T\overline{x}  \label{transformation}
\end{equation}%
into the original plant model (\ref{system}) gives us%
\begin{equation}
\overline{x}^{+}=\overline{A}\overline{x}+\overline{B}_{1}u_{1}+\overline{B}%
_{2}u_{2},  \label{reformu}
\end{equation}%
in which $\overline{A}=T^{-1}AT,\quad \overline{B}_{1}=T^{-1}B_{1},\quad
\overline{B}_{2}=T^{-1}B_{2}$. Substituting $T$ into (\ref{reformu}) gives $%
\overline{A}=diag(\overline{A}_{1},\overline{A}_{2}),$ $\overline{B}_{1}=[%
\widetilde{B}_{1},0],$ $\overline{B}_{2}=[0,\widetilde{B}_{2}],$ with $%
\overline{A}_{i}$ and $\widetilde{B}_{i}$ being defined in (\ref{ABT}). With
(\ref{transformation}), the plant state has been rearranged as $\overline{x}%
=[\overline{x}_{1},\overline{x}_{2}]$ such that%
\begin{equation}
\overline{x}_{1}^{+}=\overline{A}_{1}\overline{x}_{1}+\widetilde{B}%
_{1}u_{1},\quad \overline{x}_{2}^{+}=\overline{A}_{2}\overline{x}_{2}+%
\widetilde{B}_{2}u_{2},  \label{nominal}
\end{equation}%
where $\overline{x}_{1}=[\overline{x}_{11},\overline{x}_{21}]\in \mathbf{R}^{%
\overline{n}_{1}},$ $\overline{x}_{2}=[\overline{x}_{21},\overline{x}%
_{22}]\in \mathbf{R}^{\overline{n}_{2}}$ with $\overline{n}_{1}$ and $%
\overline{n}_{2}$ being defined in (\ref{SS}). It becomes evident that the
original centralized model (\ref{system}) has been divided into two parts
that are only influenced by either $u_{1}$ or $u_{2}$. However, important
questions remain as how (\ref{transformation}) impacts the cooperative cost
function in (\ref{localcom}). Besides, to apply the proposed method, one
does not have to restructure the plant physically, since (\ref%
{transformation}) is only to be utilized in the {local optimization
algorithms}.

\subsection{Reformulating the cost by the state transformation}

Substituting $x=T\overline{x}$ into (\ref{localcom}) gives us:%
\begin{equation}
\min\limits_{\mathbf{u}_{i}}\overline{\mathscr{V}}_{i}\text{ }s.t.\text{ (%
\ref{reformu}), }u_{i}(k)\in \mathbf{U}_{i},\text{ for }k\in \mathscr{I}%
_{0:N-1},  \label{newcost}
\end{equation}%
where, $\overline{\mathscr{V}}_{i}=\sum\limits_{k=0}^{N-1}[\overline{x}^{%
\mathrm{T}}(k)\overline{Q}\overline{x}(k)+u_{i}^{\mathrm{T}}(k)\rho
_{i}R_{i}u_{i}(k)]+\overline{x}^{\mathrm{T}}(N)\overline{P}\overline{x}(N)$,%
\begin{equation}
\overline{Q}=T^{-1}QT\geq 0,\text{ }\overline{P}=T^{-1}PT>0.  \label{QPnew}
\end{equation}%
Structures of $\overline{Q}$ and $\overline{P}$ depend on the original
structures of $Q_{i},P_{i}$ and matrix $T$. For $i,j\in \mathscr{I}_{1:2}$,
denote%
\begin{equation}
Q_{i}=\left[
\begin{array}{cc}
Q_{i,1} & Q_{i,\ast } \\
Q_{i,\ast }^{\mathrm{T}} & Q_{i,2}%
\end{array}%
\right] \geq 0,\text{ }P_{i}=\left[
\begin{array}{cc}
P_{i,1} & P_{i,\ast } \\
P_{i,\ast }^{\mathrm{T}} & P_{i,2}%
\end{array}%
\right] >0,  \label{saff}
\end{equation}%
with $Q_{i,j},P_{i,j}\in \mathbf{R}^{n_{ij}\times n_{ij}},Q_{i,\ast
},P_{i,\ast }\in \mathbf{R}^{n_{i1}\times n_{i2}}.$

\begin{remark}
\label{rmt1}Given $Q_{i}\geq 0,P_{i}>0,$ it follows that $%
Q_{i,1},Q_{i,2}\geq 0,$ $P_{i,1},P_{i,2}>0,$ for $i\in \mathscr{I}_{1:2}.$
\end{remark}

From (\ref{originalcost}), (\ref{texpress}) and (\ref{QPnew}), after matrix
manipulations, we have%
\begin{equation}
\overline{Q}=\left[
\begin{array}{cc}
\overline{Q}_{11} & \overline{Q}_{12} \\
\overline{Q}_{12}^{\mathrm{T}} & \overline{Q}_{22}%
\end{array}%
\right] \geq 0,\text{ }\overline{P}=\left[
\begin{array}{cc}
\overline{P}_{11} & \overline{P}_{12} \\
\overline{P}_{12}^{\mathrm{T}} & \overline{P}_{22}%
\end{array}%
\right] >0,  \label{weights}
\end{equation}%
where, $\overline{Q}_{12}=diag(\rho _{1}Q_{1,\ast },\rho _{2}Q_{2,\ast }),%
\overline{P}_{12}=diag(\rho _{1}P_{1,\ast },\rho _{2}P_{2,\ast })$, and for $%
i=j$, $\overline{Q}_{ij}=diag(\rho _{1}Q_{1,i},\rho _{2}Q_{2,i}),$ $%
\overline{P}_{ij}=diag(\rho _{1}P_{1,i},\rho _{2}P_{2,i}),$ in which, $%
\overline{Q}_{ij},\overline{P}_{ij}\in \mathbf{R}^{\overline{n}_{i}\times
\overline{n}_{j}},$ with $\overline{n}_{1}$ and $\overline{n}_{2}$ as
defined in (\ref{SS}). From Remark \ref{rmt1}, we have $\overline{Q}%
_{ij}\geq 0,$ $\overline{P}_{ij}>0,$ when $i=j,$ for $i,j\in \mathscr{I}%
_{1:2}$. We arrive at an alternative formulation of the optimization
problems (\ref{newcost}). We believe this formulation is insightful because
it shows clearly the impact of the individual local inputs on the global
cost. To illustrate, the optimization problem for subsystem 1 now becomes%
\begin{equation}
\min\limits_{\mathbf{u}_{1}}\overline{\mathscr{V}}_{1}\text{ }s.t.\text{ (%
\ref{reformu}), }u_{1}(k)\in \mathbf{U}_{1},\text{ }k\in \mathscr{I}_{0:N-1},
\label{finalcost}
\end{equation}%
where, $\overline{\mathscr{V}}_{1}=\overline{\mathscr{V}}_{1}^{a}+\overline{%
\mathscr{V}}_{1}^{b}+\overline{\mathscr{V}}_{1}^{c},$ with $\overline{%
\mathscr{V}}_{1}^{a}=\overline{x}_{1}^{\mathrm{T}}(N)\overline{P}_{11}%
\overline{x}_{1}(N)+\sum\limits_{k=0}^{N-1}\left[ \overline{x}_{1}^{\mathrm{T%
}}(k)\overline{Q}_{11}\overline{x}_{1}(k)+u_{1}^{\mathrm{T}}(k)\rho
_{1}R_{1}u_{1}(k)\right] $, $\overline{\mathscr{V}}_{1}^{b}=\sum%
\limits_{k=0}^{N-1}\left[ 2\overline{x}_{1}^{\mathrm{T}}(k)\overline{Q}_{12}%
\overline{x}_{2}(k)\right] +2\overline{x}_{1}^{\mathrm{T}}(N)\overline{P}%
_{12}\overline{x}_{2}(N)$, $\overline{\mathscr{V}}_{1}^{c}=\sum%
\limits_{k=0}^{N-1}\left[ \overline{x}_{2}^{\mathrm{T}}(k)\overline{Q}_{22}%
\overline{x}_{2}(k)\right] +\overline{x}_{2}^{\mathrm{T}}(N)\overline{P}_{22}%
\overline{x}_{2}(N)$. From the structures of (\ref{reformu}) and $\overline{%
\mathscr{V}}_{1}$, we know that $\overline{\mathscr{V}}_{1}^{a}$ contains
these parts that are only influenced by $\mathbf{u}_{1};$ $\overline{%
\mathscr{V}}_{1}^{b}$ collects the coupling prediction dynamics between $%
\overline{x}_{1}$ and $\overline{x}_{2}$ embedded by $\overline{Q}_{12}$ and
$\overline{P}_{12}$; $\overline{\mathscr{V}}_{1}^{c}$ is only effected by $%
\mathbf{u}_{2}$ and may be neglected when computing $\mathbf{u}_{1}$. Thus,
it is $\overline{\mathscr{V}}_{1}^{b}$ that makes the computation of $%
\mathbf{u}_{1}$ dependent on $\mathbf{u}_{2}$. This clearly raises the
question {about} how this coupling term can be dealt with to facilitate the
computation of local inputs and reduce communication burden.

\section{\label{sgscopy2}Stability Analysis}

\subsection{The centralized case}

The centralized MPC problem for (\ref{system}) with the cost (\ref%
{originalcost}), after the transformation procedure in Section 3 is applied,
is equavilent to the centralized MPC problem for (\ref{reformu}) with the
cost%
\begin{equation}
\overline{\mathscr{V}}=\sum\limits_{k=0}^{N-1}[\overline{x}^{\mathrm{T}}(k)%
\overline{Q}\overline{x}(k)+u^{\mathrm{T}}(k)Ru(k)]+\overline{x}^{\mathrm{T}%
}(N)\overline{P}\overline{x}(N),  \label{cennewcost}
\end{equation}%
with $\overline{Q}$, $\overline{P}$, and $R$ being defined in (\ref{weights}%
) and (\ref{originalcost}), respectively. To guarantee stability, the
original terminal weigthing matrix $P_{i}$ has to be selected so that $%
\overline{P}$ satisfies certain conditions. Based on Assumption \textbf{A1},
one can find $K_{i}$ such that%
\begin{equation}
\overline{A}_{K_{i}}=\overline{A}_{i}+\widetilde{B}_{i}K_{i}
\label{identity}
\end{equation}%
is Schur stable. {Given that $\mathbf{U}_{i}$ contains the origin in its
interior, for }$i\in \mathscr{I}_{1:2}${, there exists a (possibly small)}
polyhedral positively invariant set $\overline{\mathscr{X}}_{T}^{i}$ around
the origin for $\overline{x}_{i}^{+}=\overline{A}_{K_{i}}\overline{x}_{i}$
so that $K_{i}\overline{x}_{i}\in \mathbf{U}_{i}$ for all $\overline{x}%
_{i}\in \overline{\mathscr{X}}_{T}^{i}.$ Denote $K=diag(K_{1},K_{2})$ and%
\begin{equation}
\overline{A}_{K}=\overline{A}+diag(\overline{B}_{1},\overline{B}_{2})K=diag(%
\overline{A}_{K_{1}},\overline{A}_{K_{2}}).  \label{KAK}
\end{equation}%
One then has that $\overline{A}_{K}$ is Schur stable. Moreover, $K\overline{x%
}$ becomes a local stablizing controller such that $K\overline{x}\in \mathbf{%
U}_{1}\times \mathbf{U}_{2}$, for $\overline{x}\in \overline{\mathscr{X}}%
_{T}^{1}\times \overline{\mathscr{X}}_{T}^{2}$. Consider the following
centralized MPC problem%
\begin{equation}
\left\{
\begin{array}{l}
\min\limits_{\mathbf{u}}\overline{\mathscr{V}}\text{ }s.t.\text{ (\ref%
{reformu}), }\overline{x}(N)\in \overline{\mathscr{X}}_{T}^{1}\times
\overline{\mathscr{X}}_{T}^{2} \\
u(k)\in \mathbf{U}_{1}\times \mathbf{U}_{2},\text{ }k\in \mathscr{I}_{0:N-1}%
\end{array}%
\right. ,  \label{centralized}
\end{equation}%
where, $u(k)=[u_{1}(k),u_{2}(k)]$ and $\mathbf{u}$ is the input sequences
for (\ref{reformu}) along the prediction horizon. Denote $\overline{%
\mathscr{X}}_{N}$ the set of all $\overline{x}$ for which there exists a
feasible $\mathbf{u}$ to (\ref{centralized}). At each sampling instant,
assume that $K\overline{x}$ is applied after the prediction horizon as a
local controller in the set $\overline{\mathscr{X}}_{T}^{1}\times \overline{%
\mathscr{X}}_{T}^{2}$. Based on the terminal triple argument \cite[Chp. 2]%
{Madison2009}, for system (\ref{reformu}) in closed-loop with the solution
to the centralized problem (\ref{centralized}), stability can be obtained,
if the original terminal weigthing $P_{i}$ is selected so that $\overline{P}$
is the unique positive definite solution to the ARE%
\begin{equation}
\overline{A}_{K}^{\mathrm{T}}\widehat{P}\overline{A}_{K}+\overline{Q}+K^{%
\mathrm{T}}RK=\widehat{P},  \label{ARECen}
\end{equation}%
i.e., $\overline{P}=\widehat{P}>0$. Note from (\ref{weights}), each of the
four partitions of $\overline{P}$ is block diagonal (BD), while $\widehat{P}$
as the solution to (\ref{ARECen}) does not necessarily have the same
structure with $\overline{P}$. Thus, the question reduces to whether $%
\widehat{P}$ has specified structures as $\overline{P}$. This issue is
closely related to the problem of structured Lyapunov functions and the
requirement of $\widehat{P}$ having the same structure as $\overline{P}$ in (%
\ref{weights}) would only be satisfied for systems with certain properties
\cite{Boyd1989}. Without making such an assumption, an alternative
sufficient condition for closed-loop stability is presented next. Denote%
\begin{equation}
\widehat{P}=\left[
\begin{array}{cc}
\widehat{P}_{11} & \widehat{P}_{12} \\
\widehat{P}_{12}^{\mathrm{T}} & \widehat{P}_{22}%
\end{array}%
\right] >0.  \label{Phat}
\end{equation}

\begin{proposition}
\label{prop}Select $K_{i}$ as in (\ref{identity}). Assume that $K\overline{x}
$ is applied after the prediction horizon as a local controller in the set $%
\overline{\mathscr{X}}_{T}^{1}\times \overline{\mathscr{X}}_{T}^{2}$. If the
terminal weigthing matrix $P_{i}$\ in (\ref{localcost}) have been selected
such that $\overline{P}$\ satisfies the matrix inequality%
\begin{equation}
\overline{A}_{K}^{\mathrm{T}}(\overline{P}-\widehat{P})\overline{A}_{K}\leq
\overline{P}-\widehat{P},  \label{requirement1}
\end{equation}%
then the origin of the closed-loop system of (\ref{reformu}) with the
solution to (\ref{centralized}), is exponentially stable in $\overline{%
\mathscr{X}}_{N}$.
\end{proposition}

\begin{proof}
Based on standard MPC stability arguments \cite[pp.142-145]{Madison2009}, if
$\overline{K}_{i}$ is selected as in (\ref{identity}), and $K\overline{x}$
is applied after the prediction horizon as a local controller in the set $%
\overline{\mathscr{X}}_{T}^{1}\times \overline{\mathscr{X}}_{T}^{2}$, to
establish stability, we only have to show that $\vartheta _{t}(\overline{A}%
_{K}\overline{x})+\ell (\overline{x},u)-\vartheta _{t}(\overline{x})\leq 0,$
$\forall \overline{x}_{1}\in \overline{\mathscr{X}}_{T}^{1},$ where $\ell (%
\overline{x},u)=\overline{x}^{\mathrm{T}}(k)\overline{Q}\overline{x}(k)+u^{%
\mathrm{T}}(k)Ru(k),$ $\vartheta _{t}(\overline{x})=\overline{x}^{\mathrm{T}%
}(N)\overline{P}\overline{x}(N)$. The above inequality holds if $\overline{A}%
_{K}^{\mathrm{T}}\overline{P}\overline{A}_{K}+\overline{Q}+K^{\mathrm{T}%
}RK\leq \overline{P}.$ By deducting (\ref{ARECen}) from the above inequality
on both sides, we obtain (\ref{requirement1}).
\end{proof}

\subsection{The distributed case}

We proceed to establish the stability of the closed-loop system with
distributed solutions. As in the previous subsection, for $i\in \mathscr{I}%
_{1:2}$, we assume that $K_{i}$ is selected with a {polyhedral} positively
invariant set $\overline{\mathscr{X}}_{T}^{i}$ around the origin for the
system $\overline{x}_{i}^{+}=\overline{A}_{K_{i}}\overline{x}_{i}$ such that
$\overline{A}_{K_{i}}$ (\ref{identity}) is Schur stable and $K_{i}\overline{x%
}_{i}\in \mathbf{U}_{i}$ for all $\overline{x}_{i}\in \overline{\mathscr{X}}%
_{T}^{i}$. Note that, $Q_{i}$ and $P_{i}$ are generally not block diagonal
(BD), i.e., $Q_{i,\ast },P_{i,\ast }\neq 0$ in (\ref{saff}). However, from
Remark~\ref{rmt1}, we still have $\overline{Q}_{ij}\geq 0,\overline{P}%
_{ij}>0,$ when $i=j,$ for $i,j\in \mathscr{I}_{1:2},$ with $\overline{Q}%
_{ij},\overline{P}_{ij}$ defined in (\ref{weights}). A closer look at the
optimization problem (\ref{finalcost}) for computing $\mathbf{u}_{1}$
reveals that as long as $\overline{x}_{1}$ (defined in (\ref{nominal})) can
be stabilized\textit{\ }by $u_{1}$ (similarly, $\overline{x}_{2}$ stabilized%
\textit{\ }by $u_{2}$), the coupling term $\overline{\mathscr{V}}_{1}^{b}$
becomes an asymptotically vanishing weight that can be overlooked in the
local optimization. Under such conditions, to compute $\mathbf{u}_{1}$, one
only needs the information that is associated with $\overline{x}_{1}$, and
there is no need of the information of $\overline{x}_{2}$ and $\mathbf{u}%
_{2} $, e.g., there is no need to perform iterations in the local
optimizations. In such cases, the resultant input sequence $\mathbf{u}_{1}$
(and $\mathbf{u}_{2}$) is a suboptimal solution to the original cooperative
problem with cost function (\ref{originalcost}). In return, benefits of
doing so are the large amount of communication burden reduction and the
independency in computing the local inputs. In this case, the local
optimization problems become%
\begin{equation}
\left\{
\begin{array}{l}
\min\limits_{\mathbf{u}_{i}}\overline{\mathscr{V}}_{i}^{a}(\overline{x}%
_{i}(0),\mathbf{u}_{i}) \\
s.t.\text{ (\ref{nominal}), }\overline{x}_{i}(N)\in \overline{\mathscr{X}}%
_{T}^{i}\text{; }u_{i}(k)\in \mathbf{U}_{i},\text{ }k\in \mathscr{I}_{0:N-1}%
\end{array}%
\right.  \label{local}
\end{equation}%
where, $\overline{\mathscr{V}}_{i}^{a}(\overline{x}_{i}(0),\mathbf{u}_{i})$
is in the same form as $\overline{\mathscr{V}}_{1}^{a}$ in (\ref{finalcost}%
), with $\overline{P}_{ii}$ being defined in (\ref{weights}). Denote $%
\overline{\mathscr{X}}_{N}^{i}$ the set of all $\overline{x}_{i}$ for which
there exists a feasible $\mathbf{u}_{i}$ to (\ref{local}), for $i\in %
\mathscr{I}_{1:2}$. At each sampling instant, assume that $K_{i}\overline{x}%
_{i}$ is applied after the prediction horizon.

\begin{remark}
\label{rmx2}Once we select the local costs as in subsection 2.1 (see in text
before and after (\ref{originalcost})) and follow the procedure in Section 3
to apply (\ref{transformation}), the structures of the weightings in $%
\overline{\mathscr{V}}_{i}^{a}$ of (\ref{local}) will be fixed. Especially,
the terminal weighting $\overline{P}_{ii}$ is with a BD structure as
specified in (\ref{weights}).
\end{remark}

Therefore, the question reduces to whether the solution to the optimization
problem (\ref{local}), whose terminal costs have specified BD structures,
can stabilize $\overline{x}_{i}$. In other words, if the ARE%
\begin{equation}
\overline{A}_{K_{i}}^{\mathrm{T}}\widehat{P}_{i}\overline{A}_{K_{i}}+%
\overline{Q}_{ii}+K_{i}^{\mathrm{T}}\rho _{i}R_{i}K_{i}=\widehat{P}_{i},
\label{ARED}
\end{equation}%
admits a BD positive definite solution $\widehat{P}_{i}$ having the same
structure as $\overline{P}_{ii}$ in (\ref{weights}), closed-loop stability
would follow based on the terminal triple argument \cite[Chp. 2]{Madison2009}%
. However, as noted in Remark \ref{rmx2}, $\overline{P}_{ii}$ is with a BD
structure. This problem is closely related to the question of diagonal
stability and the requirement of $\widehat{P}_{i}$ having the same BD
structure as $\overline{P}_{ii}$ in (\ref{weights}) would only be satisfied
for systems with certain properties \cite{Bhaya2000}-\cite{Anderson2016}.
Alternatively, based on Proposition \ref{prop}, we have the following result.

\begin{proposition}
\label{prop2}Select $K_{i}$ as in (\ref{identity}). At each sampling
instant, assume that $K_{i}\overline{x}_{i}$ is applied after the prediction
horizon. If $P_{i}$ \ in (\ref{localcost}) have been selected such that (\ref%
{requirement1}) holds, then $\overline{P}_{ii}$ in (\ref{weights}) {satisfies%
} the matrix inequality%
\begin{equation}
\overline{A}_{K_{i}}^{\mathrm{T}}(\overline{P}_{ii}-\widehat{P}_{ii})%
\overline{A}_{K_{i}}\leq \overline{P}_{ii}-\widehat{P}_{ii};
\label{selection2}
\end{equation}%
the origin of the closed-loop system of (\ref{reformu}) with the solution to
the problem (\ref{local}) is exponentially stable in $\overline{\mathscr{X}}%
_{N}^{1}\times \overline{\mathscr{X}}_{N}^{2},$ with $\mathbf{u}_{1}$ and $%
\mathbf{u}_{2}$ being computed independently.
\end{proposition}

\begin{proof}
Given the structure of the reformulated system (\ref{reformu}), closed-loop
stability can be proved if we show that $u_{i}$ stabilizes $\overline{x}_{i}$%
. To illustrate, we show that if the original terminal weigthing matrix $%
P_{i}$ is selected such that (\ref{requirement1}) holds, then $u_{1}$
stabilizes $\overline{x}_{1}$. Note that if we ignore the terms $\overline{%
\mathscr{V}}_{1}^{b}$ and $\overline{\mathscr{V}}_{1}^{c}$ as in the
optimization problem (\ref{local}), the cost function becomes $\overline{%
\mathscr{V}}_{1}^{a}=\sum\limits_{k=0}^{N-1}\ell (\overline{x}%
_{1}(k),u_{1}(k))+\vartheta _{t}(\overline{x}_{1}(N)),$ with $\ell (%
\overline{x}_{1},u_{1})=\overline{x}_{1}^{\mathrm{T}}\overline{Q}_{11}%
\overline{x}_{1}+u_{1}^{\mathrm{T}}\rho _{1}R_{1}u_{1}$, and $\vartheta
_{t1}(\overline{x}_{1})=\overline{x}_{1}^{\mathrm{T}}\overline{P}_{11}%
\overline{x}_{1}.$ Since the local terminal controller $K_{1}\overline{x}_{1}
$ is applied beyond the prediction horizon, following standard MPC stability
arguments \cite[pp.142-145]{Madison2009}, we only have to show that $%
\vartheta _{t1}(\overline{x}_{1})$ satisfies $\vartheta _{t1}(\overline{A}%
_{K_{1}}\overline{x}_{1})+\ell (\overline{x}_{1},u_{1})-\vartheta _{t1}(%
\overline{x}_{1})\leq 0,$ $\forall \overline{x}_{1}\in \overline{\mathscr{X}}%
_{T}^{1}.$ This inequality holds if%
\begin{equation}
\overline{A}_{K_{1}}^{\mathrm{T}}\overline{P}_{11}\overline{A}_{K_{1}}+%
\overline{Q}_{11}+K_{1}^{\mathrm{T}}\rho _{1}R_{1}K_{1}\leq \overline{P}%
_{11}.  \label{relat}
\end{equation}%
Given $K$, $\overline{A}_{K}$ in (\ref{KAK}), $R$ in (\ref{originalcost}), $%
\overline{Q}$ (\ref{weights}), and the structure of $\widehat{P}$ in (\ref%
{Phat}), the ARE (\ref{ARECen}) can be rewritten as%
\begin{equation*}
\begin{array}{l}
\left[
\begin{array}{cc}
\overline{A}_{K_{1}} &  \\
& \overline{A}_{K_{2}}%
\end{array}%
\right] ^{\mathrm{T}}\left[
\begin{array}{cc}
\widehat{P}_{11} & \widehat{P}_{12} \\
\widehat{P}_{12}^{\mathrm{T}} & \widehat{P}_{22}%
\end{array}%
\right] \left[
\begin{array}{cc}
\overline{A}_{K_{1}} &  \\
& \overline{A}_{K_{2}}%
\end{array}%
\right]  \\
+\left[
\begin{array}{cc}
\overline{Q}_{11} & \overline{Q}_{12} \\
\overline{Q}_{12}^{\mathrm{T}} & \overline{Q}_{22}%
\end{array}%
\right] +\text{ }\left[
\begin{array}{cc}
K_{1} &  \\
& K_{2}%
\end{array}%
\right] ^{\mathrm{T}}\left[
\begin{array}{cc}
\rho _{1}R_{1} &  \\
& \rho _{2}R_{2}%
\end{array}%
\right]  \\
\left[
\begin{array}{cc}
K_{1} &  \\
& K_{2}%
\end{array}%
\right] =\left[
\begin{array}{cc}
\widehat{P}_{11} & \widehat{P}_{12} \\
\widehat{P}_{12}^{\mathrm{T}} & \widehat{P}_{22}%
\end{array}%
\right] .%
\end{array}%
\end{equation*}%
Taking out the the first block diagonal components in the above expression
gives us
\begin{equation}
\overline{A}_{K_{1}}^{\mathrm{T}}\widehat{P}_{11}\overline{A}_{K_{1}}+%
\overline{Q}_{11}+K_{1}^{\mathrm{T}}\rho _{1}R_{1}K_{1}=\widehat{P}_{11}.
\label{ine1}
\end{equation}%
We can also obtain a more structured expression of (\ref{requirement1}):%
\begin{equation*}
\begin{array}{l}
\left[
\begin{array}{cc}
\overline{A}_{K_{1}} &  \\
& \overline{A}_{K_{2}}%
\end{array}%
\right] ^{\mathrm{T}}\left( \left[
\begin{array}{cc}
\overline{P}_{11} & \overline{P}_{12} \\
\overline{P}_{12}^{\mathrm{T}} & \overline{P}_{22}%
\end{array}%
\right] -\left[
\begin{array}{cc}
\widehat{P}_{11} & \widehat{P}_{12} \\
\widehat{P}_{12}^{\mathrm{T}} & \widehat{P}_{22}%
\end{array}%
\right] \right)  \\
\left[
\begin{array}{cc}
\overline{A}_{K_{1}} &  \\
& \overline{A}_{K_{2}}%
\end{array}%
\right] \leq \left[
\begin{array}{cc}
\overline{P}_{11} & \overline{P}_{12} \\
\overline{P}_{12}^{\mathrm{T}} & \overline{P}_{22}%
\end{array}%
\right] -\left[
\begin{array}{cc}
\widehat{P}_{11} & \widehat{P}_{12} \\
\widehat{P}_{12}^{\mathrm{T}} & \widehat{P}_{22}%
\end{array}%
\right] .%
\end{array}%
\end{equation*}%
If the original terminal {weighting matrices} $P_{i}$ is selected such that (%
\ref{requirement1}) holds, the diagonal components of the the above
inequality must hold, i.e., $\overline{A}_{K_{1}}^{\mathrm{T}}(\overline{P}%
_{11}-\widehat{P}_{11})\overline{A}_{K_{1}}\leq \overline{P}_{11}-\widehat{P}%
_{11}$. Adding the above inequality and (\ref{ine1}) from both sides gives (%
\ref{relat}). Therefore, $u_{1}$ stabilizes $\overline{x}_{1}$. A similar
argument can be made so that $u_{2}$ stabilizes $\overline{x}_{2}$ with $%
\overline{\mathscr{V}}_{2}^{b}$ and $\overline{\mathscr{V}}_{2}^{c}$ being
ignored in the local optimization. The independency between the computation
of $\mathbf{u}_{1}$ and $\mathbf{u}_{2}$ follows from the structure of the
cost in (\ref{local}).
\end{proof}

To apply the proposed framework with stability guarantee, condition (\ref%
{requirement1}) has to be enforced as an additional requirement when
selecting $P_{i}$\ in (\ref{localcost}), apart from the requirement that%
\textbf{\ }$P_{i}>0$\textbf{.} Also, the satisfaction of condition (\ref%
{requirement1}) is a joint property of the considered system and the
weightings $Q_{i}$ and $P_{i}$. Overall, Propositions \ref{prop}-\ref{prop2}
reveal that, within the proposed framework, local inputs can be computed
independently while still guaranteeing plant-wide closed-loop stability, if
condition (\ref{requirement1}) holds. Given the independency in computing
the local inputs, there is a possibility to generalize the centralized
tuning paradigms in \cite{KongIJC}-\cite{Kong2013} to the distributed case.

\begin{remark}
As remarked earlier, $Q_{i}$ and $P_{i}$ in (\ref{originalcost}) are
generally not BD. Given the arguments in the above stability analysis, one
would be tempted to choose $Q_{i}$ and/or $P_{i}$ to be BD, as long as they
are compatible with the dimension of the state elements in (\ref{sub}). It
is worthwhile emphazing that doing so is not sensible for the considered
problem setup. The main reason for this is if both $Q_{i}$ and $P_{i}$ are
BD, by the BD structure of both the system (\ref{system}) and the global
cost (\ref{originalcost}), there will be no coupling terms at all in the
local optimization (\ref{localcom}), i.e., the distributed solution will
always be the centralized solution, and there remains little motivation to
consider such a scenario. Thus, to avoid this situation, for $i\in %
\mathscr{I}_{1:2}$, either $Q_{i}$ or $P_{i}$ has to be not BD.
\end{remark}

\subsection{Further remarks}

Although the distributed policies without iterations resulting from (\ref%
{local}) are a suboptimal solution to the original cooperative problem with
cost function (\ref{originalcost}), they are actually optimal solutions to
the following centralized problem%
\begin{equation}
\left\{
\begin{array}{l}
\min\limits_{\mathbf{u}_{1},\mathbf{u}_{2}}\overline{\mathscr{V}}^{a}\text{ }%
s.t.\text{ (\ref{nominal}), }\overline{x}(N)\in \overline{\mathscr{X}}%
_{T}^{1}\times \overline{\mathscr{X}}_{T}^{2} \\
u_{1}(k)\in \mathbf{U}_{1},u_{2}(k)\in \mathbf{U}_{2},\text{ for }k\in %
\mathscr{I}_{0:N-1},\text{ }%
\end{array}%
\right.  \label{modcent}
\end{equation}%
where, $\overline{\mathscr{V}}^{a}=\overline{\mathscr{V}}_{1}^{a}+\overline{%
\mathscr{V}}_{2}^{a}=\sum\limits_{k=0}^{N-1}\left[ \overline{x}^{\mathrm{T}%
}(k)\overline{Q}_{a}\overline{x}(k)+u^{\mathrm{T}}(k)Ru(k)\right] +\overline{%
x}^{\mathrm{T}}(N)\overline{P}_{a}\overline{x}(N)$, with $\overline{%
\mathscr{V}}_{1}^{a}$ being defined in (\ref{finalcost}) and $\overline{%
\mathscr{V}}_{2}^{a}$ defined similarly, and $\overline{Q}_{a}=diag(%
\overline{Q}_{11},\overline{Q}_{22}),$ $\overline{P}_{a}=diag(\overline{P}%
_{11},\overline{P}_{22})$. Note that both the system dynamics and the cost
functions in the above optimization problem are of a {BD} structure. Thus
this optimization problem becomes separable for the local subsystems,
rendering its solution same with that of (\ref{local}). The only difference
between the above optimization problem and the original cooperative question
with cost function (\ref{originalcost}) is in the state weightings, i.e.,
the BD entries of state weightings in (\ref{originalcost}) (see the
relationships in (\ref{QPnew})) are taken out and placed in $\overline{%
\mathscr{V}}^{a}$ in (\ref{modcent}). Denote $\widetilde{Q}_{a}=\overline{Q}-%
\overline{Q}_{a},$ $\widetilde{P}_{a}=\overline{P}-\overline{P}_{a}.$
Obviously, $\widetilde{Q}_{a}$ and $\widetilde{P}_{a}$ become symmetric
matrices with zero diagonal entries. Matrices with zero diagonal entries are
also called hollow matrices. Generally speaking, $\widetilde{Q}_{a}$ and $%
\widetilde{P}_{a}$ are indefinite, i.e., their eigenvalues might be
positive, negative and zero. Therefore, the value of the coupled part in the
global cost ($\overline{\mathscr{V}}_{1}^{b}$ in (\ref{finalcost})) may be
positive or negative. We will illustrate this point in Section \ref{example}.

It should also be noted that the proposed D\&C\ framework can be realized
with iterations as in standard cooperative frameworks. If we adopt some warm
start techniques and perform iterations in the local optimization as in
standard cooperative frameworks, closed-loop stability of the D\&C framework
with iterations can be obtained if the terminal weighting $P_{i}$\ in (\ref%
{localcost}) have been selected such that (\ref{requirement1}) holds (this
can be proved by following arguments in Proposition \ref{prop} and \cite[%
Theorem 9]{Rawlings2010}, and we do not present further details here). An
important observation to be made here is that closed-loop stability for the
system (\ref{reformu}) with the centralized solution, the distributed
solution without or with iterations, respectively, can be guaranteed
simultaneously by selecting the terminal weigthing matrix $P_{i}$\ in (\ref%
{localcost}) so that (\ref{requirement1}) holds, i.e., one can change the
implementation method during operation without affecting the closed-loop
stability.

We next present a sufficient condition for (\ref{requirement1}) to hold. A
matrix $A\in \mathbf{R}^{n\times n}\ $is diagonally dominant (DD) if $%
\left\vert a_{ii}\right\vert \geq \sum\limits_{j\neq i}\left\vert
a_{ij}\right\vert $, for $i=1,...,n$ \cite[pp. 392]{Horn2013}.

\begin{corollary}
\label{cor1first} Denote $\widehat{P}$ as the solution to the ARE (\textit{%
\ref{ARECen}). }If the terminal weigthing matrix $P_{i}$\ in (\ref{localcost}%
) have been selected such that $\overline{P}-\widehat{P}-\overline{A}_{K}^{%
\mathrm{T}}(\overline{P}-\widehat{P})\overline{A}_{K}$\ is DD with
nonnegative diagonal entries, then condition (\ref{requirement1})\ holds.
\end{corollary}

\begin{proof}
From \cite[pp. 15]{Berman2003}, a DD and symmetric matrix with nonnegative
diagonal entries is positive semidefinite. Given $\widehat{P}$ as the
solution to the ARE (\ref{ARECen}\textit{)}, if the original terminal
weigthing matrix $P_{i}$ in (\ref{localcost}) is selected so that $\overline{%
P}-\widehat{P}-\overline{A}_{K}^{\mathrm{T}}(\overline{P}-\widehat{P})%
\overline{A}_{K}$ is DD with nonnegative diagonal entries, the former matrix
is positive semidefinite. Therefore, condition (\ref{requirement1}) holds.
\end{proof}

\section{\label{sgscopy3}Some Generalizations and Discussions}

With $M$ subsystems, the plant-wide system becomes:%
\begin{equation}
x^{+}=Ax+\sum\limits_{i\in \mathscr{I}_{1:M}}B_{i}u_{i},
\label{overallplant}
\end{equation}%
where, $x=[x_{1},\cdots ,x_{M}],$ $A=diag(A_{1},\cdots ,A_{M}),$ $B_{i}=[%
\overline{B}_{1i},\cdots ,\overline{B}_{Mi}],$ for $i\in \mathscr{I}_{1:M},$
with $A_{i}$ and $B_{i}$ in similar structures with $A_{1}$ and $B_{1}$ in (%
\ref{localsystem}) and (\ref{system}), respectively. Denote $\mathbf{u=[%
\mathbf{u}_{1},\cdots ,\mathbf{u}_{M}]}$. The global cost is $\mathscr{V}%
(x(0),\mathbf{u})=\sum\limits_{i\in \mathscr{I}_{1:M}}\rho
_{i}V_{i}(x_{i}(0),\mathbf{u}_{i})$, where $V_{i}(x_{i}(0),\mathbf{u}_{i})$
takes the form of $V_{i}$ in subsection 2.1. To compute $\mathbf{u}_{i}$,
each local system solves
\begin{equation}
\min\limits_{\mathbf{u}_{i}}\mathscr{V}_{i}(x(0),\mathbf{u})\text{ }s.t.%
\text{ (\ref{overallplant}), }u_{i}(k)\in \mathbf{U}_{i},\text{ for }k\in %
\mathscr{I}_{0:N-1},  \label{withoutiterations}
\end{equation}%
where $\mathscr{V}_{i}(x(0),\mathbf{u})$ is with a similar structure to that
of $\mathscr{V}_{i}$ in (\ref{localcom}). We propose the transformation $%
x=T_{M}\overline{x}$, where%
\begin{equation}
T_{M}=\left[
\begin{array}{cccc}
T_{_{1,1}} & T_{_{1,2}} & \cdots & T_{_{1,M}} \\
T_{_{2,1}} & T_{_{2,2}} & \cdots & T_{_{2,M}} \\
\vdots & \ddots & \ddots & \vdots \\
T_{_{M,1}} & \cdots & T_{_{M,M-1}} & T_{_{M,M}}%
\end{array}%
\right] \in \mathbf{R}^{\widehat{n}\times \widehat{n}},  \label{TM}
\end{equation}%
and, $\widehat{n}=\sum\limits_{i\in \mathscr{I}_{1:M}}\sum\limits_{j\in %
\mathscr{I}_{1:M}}n_{ij};$ $T_{_{i,i}}\in \mathbf{R}^{\widetilde{n}%
_{i}\times \overrightarrow{n}_{i}},$ with $\widetilde{n}_{i}=\sum\limits_{j%
\in \mathscr{I}_{1:M}}n_{ij}\ $and $\overrightarrow{n}_{i}=\sum\limits_{j\in %
\mathscr{I}_{1:M}}n_{ji},$ is a BD matrix with its $i$-th block entry as an
identity matrix of dimension $n_{ii}$ and other entries as zeros, for $i\in %
\mathscr{I}_{1:M};$ for $i,j\in \mathscr{I}_{1:M}$ and $i<j$, $T_{_{i,j}}$
is a block partitioned matrix with its $(j,i)$ entry as an identity matrix
of dimension $n_{ij}$ and other block entries as zeros; for $i,j\in %
\mathscr{I}_{1:M}$ and $i>j$, $T_{_{i,j}}$ is a block partitioned matrix
with its $(j,i)$ entry as an identity matrix of dimension $n_{ij}$ and other
block entries as zeros. Moreover, it can be verified that $T_{M}^{\mathrm{T}%
}T_{M}=I_{\widehat{n}}$, i.e., $T_{M}^{-1}=T_{M}^{\mathrm{T}}.$ The impact
of the transformation $T_{M}$ on the system (\ref{overallplant}) and the
cost functions (\ref{withoutiterations}) can be conducted as in Section \ref%
{sgscopy}. Stability analysis can also be established as in Section \ref%
{sgscopy2}.

Moreover, the results in Sections \ref{sgscopy}-\ref{sgscopy2} for the case
with state feedback can be extended to the case with output feedback by {%
embellishing} the ideas in Sections \ref{sgscopy}-\ref{sgscopy2} with
decentralized estimation (based on Assumption A3). Then the problem at hand
becomes robust stability analysis in the presence of bounded (decaying)
disturbances, and can be dealt with by applying robust MPC techniques such
as tube-based MPC {or by exploiting the inherent robustness properties of
suboptimal MPC \cite{Pannocchia2011}.} We do not present detailed
discussions here due to limited space.

A summary of the comparisons between the proposed method and the technique
in \cite{Rawlings2010} is presented in the Table~\ref{tb1:comparisonrawlings}%
. It can be concluded that the proposed method, {without} iterations,
requires lower information and communication overhead than the technique in
\cite{Rawlings2010}, {while still enjoying the advantages of the stability
properties, as well as convergence to the global centralized performance,
when implemented with iterations.}
\begin{table}[t]
\caption{Comparisons between the two frameworks}
\label{tb1:comparisonrawlings}\renewcommand{\arraystretch}{1.3} \centering
\par
\begin{tabular}{ccc}
\hline
Method & Information need & Realization \\ \hline
\cite{Rawlings2010} & plant-wide state and input & with iterations \\
D\&C & partial plant-wide state & without or \\
& (without iterations); same as & with iterations \\
& \cite{Rawlings2010} (with iterations) &
\end{tabular}%
\end{table}
It will be hard to make a comprehensive comparison between the proposed
method and many other existing techniques. Despite the potential advantages
of the proposed method, it has been noted already that the adopted model
might be non-minimal. We also remark that the assessment of a specific
method should not only be based on the method itself but also the particular
targeted application. Moreover, only the case of decoupled inputs is
considered in this paper. Whether and how the proposed framework can deal
with coupled input constraints remains as an interesting future topic.
Another topic of interest is to extend the method to linear programming
(LP), for which some preliminary results have been presented in \cite%
{KongIFAC}.

\section{\label{example}Illustrative Examples}

\subsection{An academic example with three subsystems}

The first example is presented to show the controller design procedure of
the proposed framework. The considered system is composed of three
subsystems. The local systems in (\ref{overallplant}) are defined by $%
A_{i}=diag(A_{i1},A_{i2},A_{i3}),$ $\overline{B}_{i1}=[B_{i1},0,0]$, $%
\overline{B}_{i2}=[0,B_{i2},0]$, $\overline{B}_{i3}=[0,0,B_{i3}]$, in which%
\begin{equation*}
\begin{array}{l}
A_{11}=A_{12}=A_{13}=\left[
\begin{array}{cc}
0.6 & 0.5 \\
0.1 & 0.4%
\end{array}%
\right] ,\text{ }B_{11}=B_{12}=\left[
\begin{array}{c}
1 \\
0%
\end{array}%
\right] , \\
B_{13}=\left[
\begin{array}{c}
1 \\
-0.5%
\end{array}%
\right] ,\text{ }A_{21}=A_{22}=A_{23}=\left[
\begin{array}{cc}
0.2 & 0.1 \\
0.1 & 1%
\end{array}%
\right] , \\
B_{21}=\left[
\begin{array}{c}
0.5 \\
1%
\end{array}%
\right] ,\text{ }B_{22}=\left[
\begin{array}{c}
0.6 \\
0.8%
\end{array}%
\right] ,\text{ }B_{23}=\left[
\begin{array}{c}
0.6 \\
0.9%
\end{array}%
\right] , \\
A_{31}=A_{32}=A_{33}=\left[
\begin{array}{cc}
-0.1 & 0.6 \\
1 & -0.2%
\end{array}%
\right] ,\text{ }B_{31}=B_{32}=\left[
\begin{array}{c}
0 \\
1%
\end{array}%
\right] , \\
B_{32}=\left[
\begin{array}{c}
0 \\
1%
\end{array}%
\right] ,B_{33}=\left[
\begin{array}{c}
1 \\
0.8%
\end{array}%
\right] .%
\end{array}%
\end{equation*}%
Since we have $n_{ij}=2,$ for $i,j\in \mathscr{I}_{1:3},$ the state
transformation matrix is $T=\left[
\begin{array}{ccc}
T_{_{1,1}} & T_{_{1,2}} & T_{_{1,3}} \\
T_{_{1,2}}^{\mathrm{T}} & T_{_{2,2}} & T_{_{2,3}} \\
T_{_{1,3}}^{\mathrm{T}} & T_{_{2,3}}^{\mathrm{T}} & T_{_{3,3}}%
\end{array}%
\right] \in \mathbf{R}^{18\times 18},$ where, $T_{_{i,j}}$ with $i,j\in %
\mathscr{I}_{1:3}$ is generated from the structure of (\ref{TM}). The three
local inputs are assumed to satisfy $u_{i}(k)\in \left[ -4,4\right] ,$ for $%
i\in \mathscr{I}_{1:3}.$ The original local stage cost of each subsystem is
defined by $R_{1}=R_{2}=1,$ $R_{3}=0.5,$ $Q_{2}=2Q_{1},$ $Q_{3}=0.1Q_{1}$,
where, $Q_{1}=\left[
\begin{array}{ccc}
10I & 2I & 3E \\
2I & 15I & 3E \\
3E & 3E & 20I%
\end{array}%
\right] ,$ with $E$ standing for square matrices having all its entries as
1, and $P_{i}$ is to be designed. The relative weights of the local cost
functions are $\rho _{1}=1,$ $\rho _{2}=0.5$, $\rho _{3}=1.$ We follow the
procedure in Section 3 and apply the state transformation to the original
system and the global cost function. For each open-loop unstable virtual
subsystem, a local stabilizing LQR controller is designed with the following
weights $Q_{L_{1}}=10I,Q_{L_{2}}=5I,Q_{L_{3}}=0.2I;$ $%
R_{L_{1}}=R_{L_{2}}=0.1,R_{L_{3}}=0.01.$ It can be verified that the unit
ball of $6$ dimensions around the origin becomes a positive invariant set
for the three virtual subsystems in closed-loop with their local LQR
controllers, respectively. Therefore, this set has been enforced as a
terminal constraint in the local optimization problems (\ref{local}). The
terminal costs $P_{i}$ in (\ref{localcost}) are selected so that the
condition (\ref{requirement1}) and the requirement that $P_{i}\in \mathbf{R}%
^{6\times 6}>0$ hold simultaneously. The prediction horizon is chosen to be $%
N=8$.

(i) We firstly compared the distributed solution without iterations and the
centralized solution. The initial states of the original subsystems are
taken to be $x_{1}=\left[
\begin{array}{cccccc}
-10 & -4 & 9 & 7 & 8 & 5%
\end{array}%
\right] ^{\mathrm{T}},$ $x_{2}=\left[
\begin{array}{cccccc}
-8 & -5 & 7 & 3 & 3 & 6%
\end{array}%
\right] ^{\mathrm{T}},$ $x_{3}=\left[
\begin{array}{cccccc}
-5 & -6 & 8 & -9 & 8 & 3%
\end{array}%
\right] ^{\mathrm{T}}$. At each sampling instant, we compute the value of
the global cost (GC) with the distributed solutions without iterations and
the centralized solution, respectively. Note that when computing the
distributed solutions without iterations, the coupled parts in the global
cost (as in (\ref{local}) for the case with 2 subsystems) are ignored. To
better compare the distributed solution without iterations against the
centralized solution, we will also compare the cost value of the coupled
parts with either solution, respectively. To make a fair comparison, the
state information for both centralized and distributed optimization problems
at each sampling instant is updated by the distributed solution, i.e., the
states are enforced to be the same for both optimization problems and it is
the ways the two solutions being computed make the difference in the value
of the GC and the coupled parts. The simulation results are illustrated in
Figure 1. It can be seen from Figure 1 that the distributed solution always
renders slightly worse performance than the centralized solution for the GC.
Nonetheless, the global performance loss is very moderate and remains within
$0.5\%$. Regarding the coupled cost (CC), it can be observed that its value
can change sign due to the discussions made after (\ref{modcent}). The value
of the coupled cost with the distributed solution is always larger than that
with the centralized solution. This is not unexpected, since the coupled
part has been included in the centralized optimization but not in the
distributed optimizations.
\begin{figure}[tbph]
\centering
\includegraphics[width=0.4\textwidth,bb=5 205 580
640]{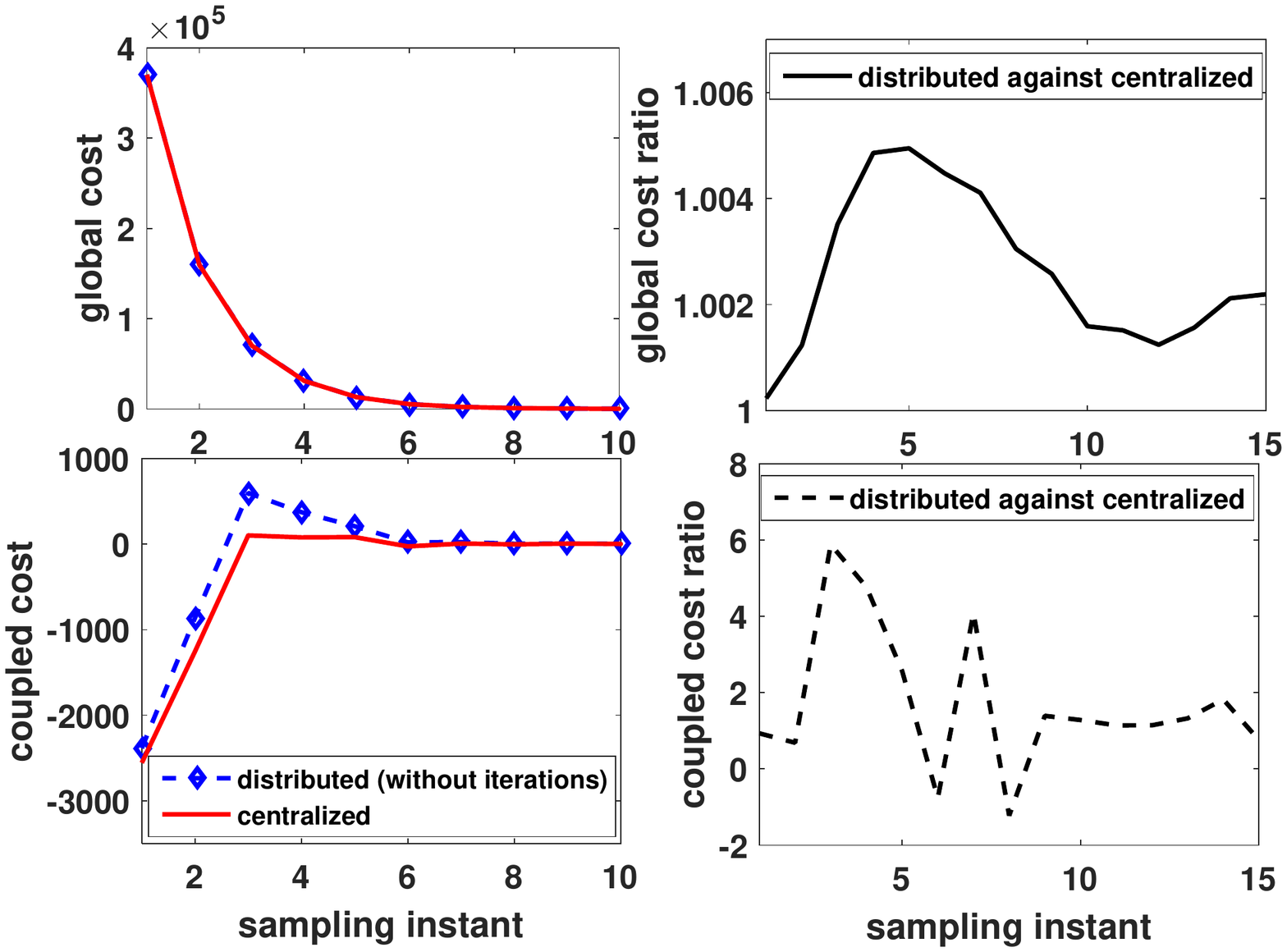}
\caption{Value evolution of the global cost and the coupled part with
centralized (red solid line) and distributed inputs without iterations (blue
dashed lines with diamonds), and their respective ratio (distributed against
centralized)}
\end{figure}
Note that in the above simulation, the distributed solutions are computed
via optimization problems (\ref{local}), i.e., the coupled parts $\overline{%
\mathscr{V}}_{i}^{b}$and $\overline{\mathscr{V}}_{i}^{c}$ in (\ref{finalcost}%
) have not been considered. However, the values of the CC have been added
when calculating the value of the GC with distributed solutions. To
understand better the performance loss of the distributed solution without
iterations, we draw 1000 sets of initial states whose entries are random
variables between $[-8,8].$ The average global performance loss from these
1000 random initial conditions is obtained to be within $0.5\%$ and the
worse global performance loss is less than $2\%.$ Note that compared to the
centralized solution, the local distributed solution without iterations only
needs $50\%$ (for this particular example) of the global state information,
and does not require the information of other local inputs.

(ii) We next compare the distributed solution without iteration, standard
cooperative solution (with iterations) and the centralized solution. For
this purpose, we take the initial conditions of the original subsystems to
be $x_{1}=\left[
\begin{array}{cccccc}
10 & 10 & 8 & 6 & -6 & 6%
\end{array}%
\right] ^{\mathrm{T}}$, $x_{2}=\left[
\begin{array}{cccccc}
10 & 2 & 3 & 5 & 3 & 6%
\end{array}%
\right] ^{\mathrm{T}}$, $x_{3}=\left[
\begin{array}{cccccc}
6 & -4 & 4 & 2 & 2 & 3%
\end{array}%
\right] ^{\mathrm{T}}$. For the first three sampling instants we use the
distributed solution without iterations to initialize the optimization
procedure. For the three strategies, we use the state information at the
fourth sampling instant as the initial conditions, i.e., the initial
condition is chosen to be the same for three different methods. Also, the
shifted input sequences of the distributed solution without iterations are
used as warm starts in obtaining the standard cooperative distributed MPC
solution with iterations.
\begin{table}[t]
\caption{Cost comparison of different strategies}
\label{tb2:comparison}\renewcommand{\arraystretch}{1.3} \centering%
\begin{tabular}{ccccc}
\hline
Method & GC & GC loss & CC & CC loss \\ \hline
Centralized & 5.1029$\cdot 10^{3}$ & 0 & 178.3 & 0 \\
5 iters & 5.1151$\cdot 10^{3}$ & $0.24\%$ & $216.5$ & $21.46\%$ \\
4 iters & 5.1218$\cdot 10^{3}$ & $0.37\%$ & $221.1$ & $24.01\%$ \\
3 iters & 5.1383$\cdot 10^{3}$ & $0.69\%$ & $228.3$ & $28.04\%$ \\
2 iters & 5.2018$\cdot 10^{3}$ & $1.94\%$ & $243.5$ & $36.56\%$ \\
1 iter & 5.5912$\cdot 10^{3}$ & $9.57\%$ & $289.2$ & $62.20\%$ \\
no iters & 5.1671$\cdot 10^{3}$ & $1.26\%$ & $303.8$ & $70.38\%$%
\end{tabular}%
\end{table}
We compare the value of the GC and CC associated with distributed solution
(without and with iterations) against the centralized solution. The results
are summarized in the Table 2. From Table 2, it can be seen that the value
of GC with the distributed solution without iterations is smaller than that
of standard cooperative distributed solution with 1 iteration. Moreover, the
strategy with iterations outperforms the case without iterations in that the
value of the CC with iterations is smaller. When more iterations are
conducted, the performance of the solution with iterations improves and
becomes better than the solution without iterations. Besides, as more
iterations are performed, the performance of the distributed strategy
improves, and with just 5 iterations, its global performance loss reduces to
less than $0.5\%$. One can also verify that when a large number of
iterations is conducted, the standard cooperative distributed solution
converges to the centralized solution.
\begin{figure}[tbph]
\centering
\includegraphics[width=0.4\textwidth,bb=65 305 575 605]{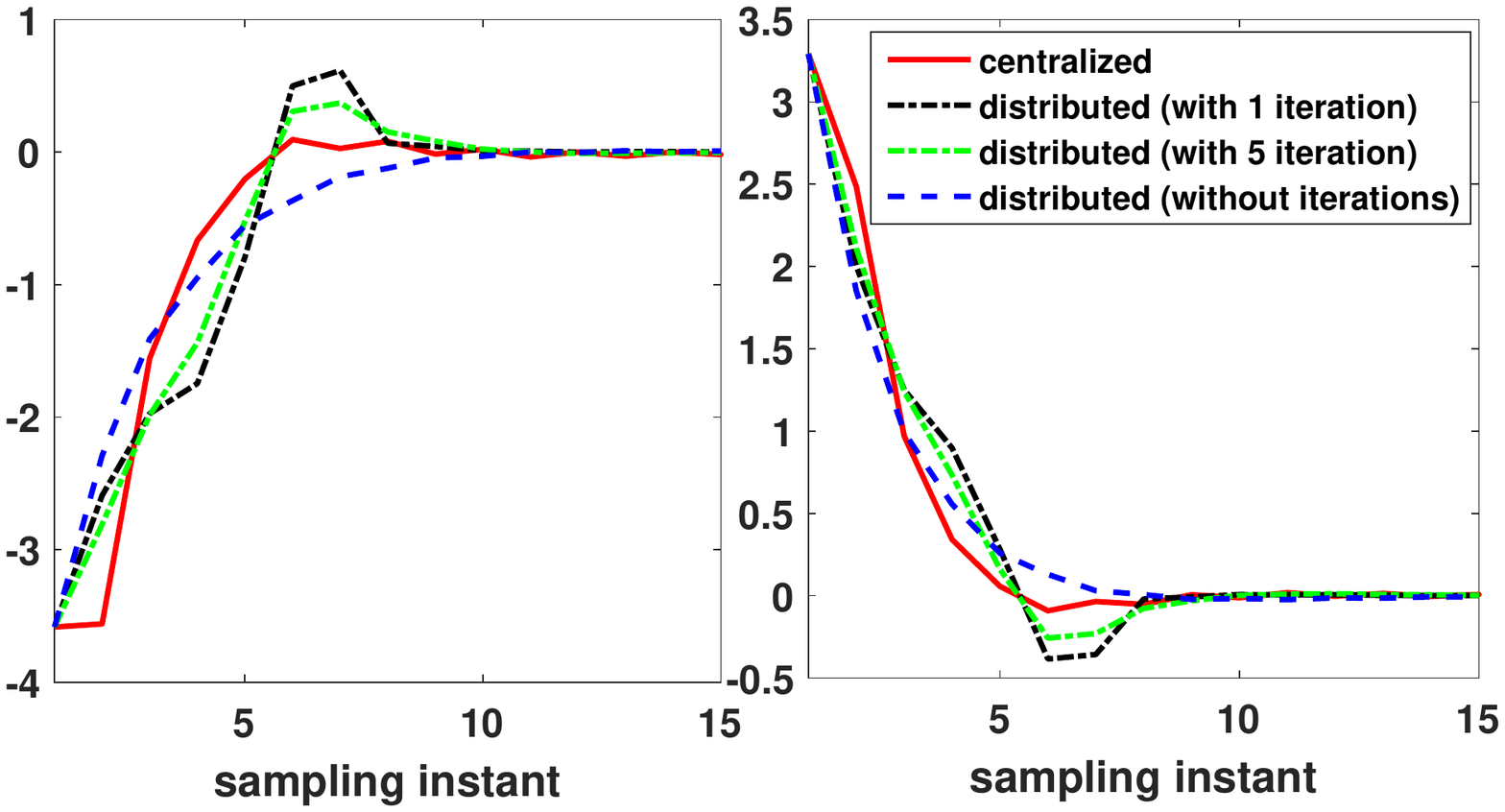}
\caption{Closed-loop trajectories of $x_{23}$ with different strategies}
\end{figure}

(iii) Finally, we compare the closed-loop trajectories of the system with
different control strategies, respectively, using the initial conditions of
the second simulation (ii). For the first sampling instant, we use
distributed solution without iterations to initialize the optimization
procedure. From the second sampling instant, the state information at each
sampling instant is updated by different strategies, respectively. In Figure
2, we show the closed-loop trajectories of $x_{23}$ with different control
strategies. It can be seen that the centralized solution offers the best
performance among all the strategies; the standard cooperative distributed
solution with 5 iterations can render a slightly quicker response than the
solution without iterations and with 1 iteration. Also, the response of the
solution with 1 iteration is slightly slower than that of the solution
without iterations.
\begin{table}[t]
\caption{Execution time comparison of different strategies}%
\renewcommand{\arraystretch}{1.3} \centering%
\begin{tabular}{ccc}
\hline
Method & Worse case & Average \\ \hline
Centralized & 0.0237s & 0.0085s \\
5 iters & 0.0373s & 0.0257s \\
1 iter & 0.0069s & 0.0055s \\
no iters & 0.0237s & 0.0094s%
\end{tabular}%
\end{table}
We have run this set of simulation for 200 sampling instants and solved the
QPs using quadprog in Matlab on a desktop computer with the processor Intel
Xeon Processor E5-2650 v2 (processor base frequency 2.6 Ghz). A summary
comparison of the computation time of different strategies is shown in Table
3. Note that we choose the worse case scenario among all the local inputs
for calculating the execution time for the solutions with iterations at each
sampling instant. To be more specific, for solutions without iterations and
with 1 iterations, if the execution time at a certain sampling instant, for
the three local inputs is $t_{1}$, $t_{2}$, $t_{3}$, respectively, we then
choose the largest one of the three as the computation time for the current
sampling instant. For the standard cooperative distributed solution with 5
iterations, the execution time at a certain iteration within a certain
sampling interval, for the three local inputs is $t_{ij}$ with $i\in %
\mathscr{I}_{1:3}$, $j\in \mathscr{I}_{1:5}$, we choose the largest value of
$\sum\limits_{j=1}^{5}t_{ij}$, given $i\in \mathscr{I}_{1:3}$, as the
computation time for the current sampling instant. A similar procedure is
carried out at each sampling instant, the worse case and the averge
computation time in Table 3 is then calculated from the record for the 200
simulations. From Figure 2 and Table 3, we can see that the standard
cooperative distributed solution with iterations generally takes more time
to compute with potentially improved performance. It should also be noted
that the worse case computation time for the centralized solution is not
significantly different from that of the solutions without iterations or
with 5 iterations. This is not very surprising given the moderate complexity
of the example. Also, as it has been remarked in \cite{Rawlings2010}, a
major argument to develop distributed methods, as opposed to a centralized
solution, is often not computational, but organizational, i.e., in the
latter case, all subsystems rely on a central decision maker to coordinate
and maintain plant-wide actions, leading to organizational inefficiency,
implementation and maintenance difficulties.

Based on all the previous simulations, we conclude that the distributed
solution without and with iterations have their respective advantages. We
suggest that one should adopt a combination of these two implementation
methods. For example, one could use the solution without iterations to
initialize the distributed algorithm; and perform iterations for certain
steps in the transient process to potentially reduce the coupled costs as
well as improve the global performance; as the states approach the origin,
one could switch back to the solution without iterations. Note that as long
as the local weights are selected so that the condition (\ref{requirement1})
holds, a change of the implementation method during operation would not
affect closed-loop stability.

\subsection{Cooperative distributed control of a four-wheel vehicle}

To see the applicability of the proposed method to a practical problem, we
apply the technique to a four wheel drive (4WD) vehicle dynamics problem
with four independent wheel motors \cite{siampisVSD14}. Whilst the
application of MPC to vehicle dynamics problems is not new, the main
motivation of examining the cooperative distributed MPC technique here is to
potentially improve the vehicle stability by distributedly controlling the
torques to four independent wheel motors while considering their respective
impact on the whole vehicle. Note that the vehicle and tyre models used in
this example were originally proposed in \cite{Velenis2011} and are very
similar to those in \cite{siampisVSD14}. The interested reader can refer to
\cite{siampisVSD14}-\cite{Velenis2011} for more technical details such as
the description of the equations of motion, tyre models, etc. The vehicle
under consideration in this example is a small sports car with parameters
reported in \cite{siampisVSD14}. In order to control four wheel motors in a
distributed way, we partition the centralized discrete-time model into 4
subsystems by assuming that the information of the vehicle velocity $V$, the
sideslip angle $\beta $, the yaw rate $\dot{\psi}$, and the four motors'
respective torque information is available at each local subsystem via the
vehicle bus communication network. We also assume that a constant steering
input of $\delta ^{ss}=10deg$ is applied on both the front wheels. To apply
the proposed cooperative distributed MPC technique, we linearize the vehicle
dynamics model around the steady state (SS) cornering condition $%
V^{ss}=11m/s,\beta ^{ss}=0.717deg,\dot{\psi}^{ss}=44.45deg,$ and derive a
centralized linear discrete-time model with sampling time 0.1 s. The
prediction horizon is selected to be {$N=8$}. The wheel torque is required
to stay within the specified motor torque limits of $\left[ -400,400\right] $%
.
\begin{figure}[tbph]
\centering
\includegraphics[width=0.25\textwidth,bb=95 210 440
680]{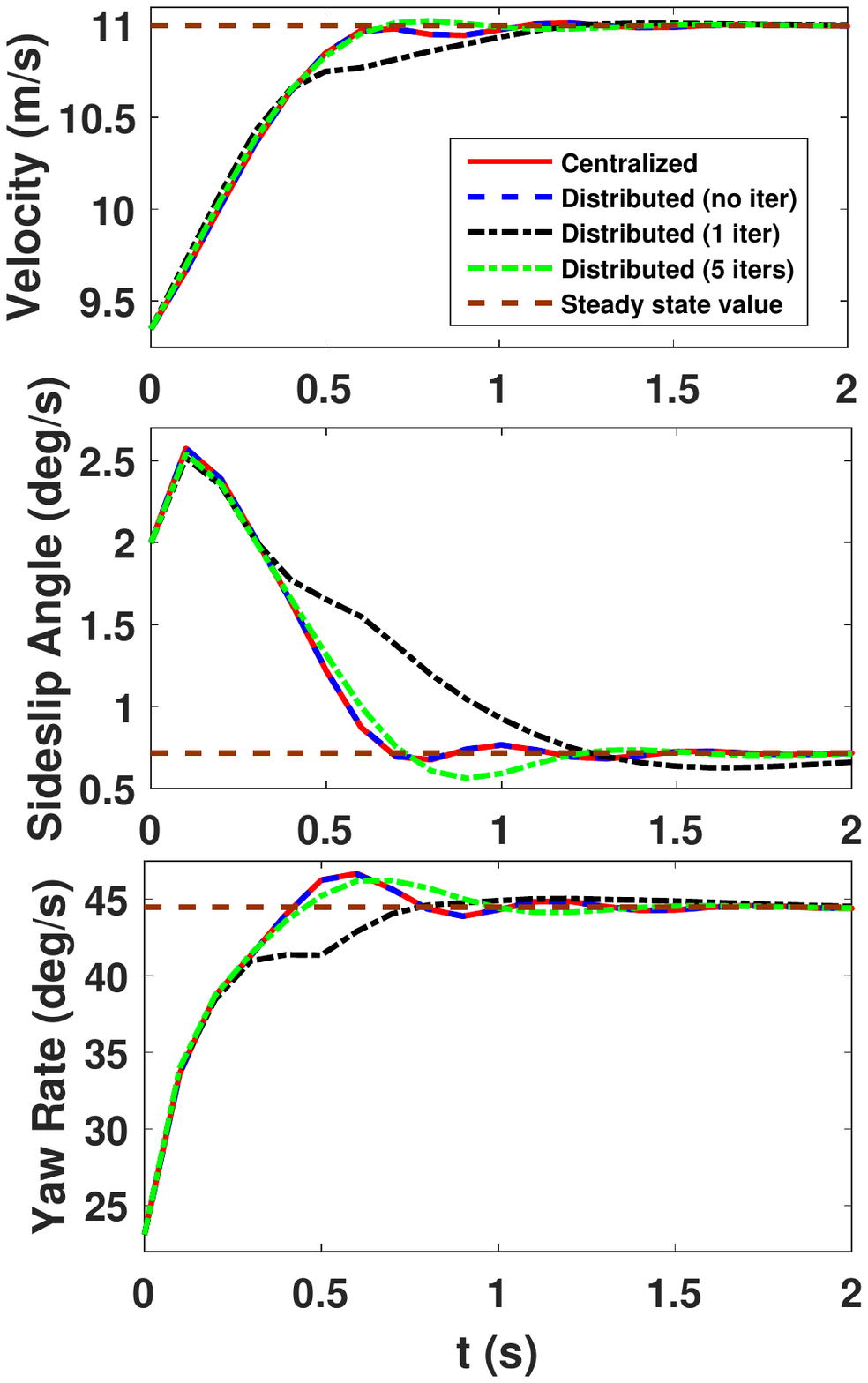}
\caption{Velocity, sideslip angle, yaw rate response with different
strategies}
\end{figure}
We compare four different strategies, i.e., centralized, distributed (no
iterations), distributed (1 iterations), and distributed (5 iterations),
against each other with the original nonlinear vehicle dynamics model in
closed-loop. For the first step, with a given state information, we use the
distributed solution without iterations to initialize the whole procedure.
Starting from the second sampling instant, the closed-loop state information
is updated by the four different methods, respectively. The responses of
vehicle velocity, sideslip angle and yaw rate are shown in Figure 3. It can
be seen from Figure 3 that the four different methods all offer fairly quick
responses. More closely, the distributed solution with 1 iteration renders
comparable but slightly worse performance than that of centralized,
distributed without iterations, and distributed with 5 iterations. We notice
similar patterns in the wheel speed response and due to limited space, these
results are not shown here. Note that in Figure 3, the vehicle dynamics
response with the centralized method and the distributed solution without
iterations are close to, but not the same with, each other.

\section{\label{conclusion}{Conclusions}}

This paper has presented a Divide and Conquer approach to the design of
cooperative distributed MPC. By recognizing the inherent structure of the
problem setup in {distributed MPC}, we propose to apply a state
transformation to the original cooperative {problem} so that the coupling
effects in the original problem setup can be dealt with more effectively.
Implications of the state transformation and sufficient conditions for
closed-loop stability are thoroughly discussed. For the case without
iterations, the proposed method allows one to compute the local inputs
independently from each other with partial plant-wide state information,
thereby potentially saving a large amount of communication overhead. The
framework can also be implemented with iterations, thereby keeping the
merits of the standard cooperative techniques with iterations. Starting from
the case of 2 systems, we have also presented the generalization to $M$
systems{, and two numerical application examples which show the benefits of
the proposed framework.}


\section*{Acknowledgment}

The authors thank Dr. Efstathios Velenis at Cranfield University, for kindly
providing the original code in \cite{Velenis2011} which we used to complete
the simulation work in the second example.





\end{document}